\documentclass{wscpaperproc}
\usepackage{latexsym}
\usepackage{graphicx}
\usepackage{mathptmx}
\usepackage[T1]{fontenc}

\usepackage{amsmath}
\usepackage{amsfonts}
\usepackage{amssymb}
\usepackage{amsbsy}
\usepackage{amsthm}

\usepackage[skip=0pt]{caption}

\usepackage{subcaption}
\usepackage{diagbox}
\def\R{{\mathbb{R}}}   

\def\N{{\mathbb{N}}}
\def\E{{\mathbb{E}}}
\def\P{{\mathbb P}}
\def\1{{\mathbf{1}}}

\def\B{{\mathcal B}}
\def\HH{{\mathbb H}}   
\def\sa{{\mathcal K}}

\usepackage[pdftex,colorlinks=true,urlcolor=blue,citecolor=black,anchorcolor=black,linkcolor=black]{hyperref}

\newtheoremstyle{wsc}
{3pt}
{3pt}
{}
{}
{\bf}
{}
{.5em}
{}

\theoremstyle{wsc}
\newtheorem{theorem}{Theorem}

\newtheorem{definition}{Definition}

\newtheorem{assumption}{Assumption}

\newtheorem{lemma}{Lemma}

\usepackage{cleveref}

\crefname{assumption}{assumption}{assumptions}

\begin{document}

%
%

\pagestyle{fancyplain}

\thispagestyle{plain}
\firstPageHead{}

\chead{\fancyplain{}{\itshape Ren, Fu, and Marcus}}

\rhead{}
\cfoot{}
\renewcommand{\headrulewidth}{0pt} 

\makeatletter
\let\@internalcite\cite
\def\cite{\def\@citeseppen{-1000}%
    \def\@cite##1##2{(##1\if@tempswa , ##2\fi)}%
    \def\citeauthoryear##1##2##3{##1 ##3}\@internalcite}
\def\citeNP{\def\@citeseppen{-1000}%
    \def\@cite##1##2{##1\if@tempswa , ##2\fi}%
    \def\citeauthoryear##1##2##3{##1 ##3}\@internalcite}
\def\citeN{\def\@citeseppen{-1000}%
    \def\@cite##1##2{##1\if@tempswa, ##2)\else{}\fi}%
    \def\citeauthoryear##1##2##3{##1 (##3)}\@citedata}
\def\citeA{\def\@citeseppen{-1000}%
    \def\@cite##1##2{(##1\if@tempswa , ##2\fi)}%
    \def\citeauthoryear##1##2##3{##1}\@internalcite}
\def\citeANP{\def\@citeseppen{-1000}%
    \def\@cite##1##2{##1\if@tempswa , ##2\fi}%
    \def\citeauthoryear##1##2##3{##1}\@internalcite}
\def\shortcite{\def\@citeseppen{-1000}%
    \def\@cite##1##2{(##1\if@tempswa , ##2\fi)}%
    \def\citeauthoryear##1##2##3{##2 ##3}\@internalcite}
\def\shortciteNP{\def\@citeseppen{-1000}%
    \def\@cite##1##2{##1\if@tempswa , ##2\fi}%
    \def\citeauthoryear##1##2##3{##2 ##3}\@internalcite}
\def\shortciteN{\def\@citeseppen{-1000}%
    \def\@cite##1##2{##1\if@tempswa, ##2\else{}\fi}%
    \def\citeauthoryear##1##2##3{##2 (##3)}\@citedata}
\def\shortciteA{\def\@citeseppen{-1000}%
    \def\@cite##1##2{(##1\if@tempswa , ##2\fi)}%
    \def\citeauthoryear##1##2##3{##2}\@internalcite}
\def\shortciteANP{\def\@citeseppen{-1000}%
    \def\@cite##1##2{##1\if@tempswa , ##2\fi}%
    \def\citeauthoryear##1##2##3{##2}\@internalcite}
\def\citeyear{\def\@citeseppen{-1000}%
    \def\@cite##1##2{(##1\if@tempswa , ##2\fi)}%
    \def\citeauthoryear##1##2##3{##3}\@citedata}
\def\citeyearNP{\def\@citeseppen{-1000}%
    \def\@cite##1##2{##1\if@tempswa , ##2\fi}%
    \def\citeauthoryear##1##2##3{##3}\@citedata}
%
%
%
\def\@citedata{%
    \@ifnextchar [{\@tempswatrue\@citedatax}%
                  {\@tempswafalse\@citedatax[]}%
}

\def\@citedatax[#1]#2{%
\if@filesw\immediate\write\@auxout{\string\citation{#2}}\fi%
  \def\@citea{}\@cite{\@for\@citeb:=#2\do%
    {\@citea\def\@citea{, }\@ifundefined
       {b@\@citeb}{{\bf ?}%
       \@warning{Citation `\@citeb' on page \thepage \space undefined}}%
{\csname b@\@citeb\endcsname}}}{#1}}%

%
\def\@citex[#1]#2{%
\if@filesw\immediate\write\@auxout{\string\citation{#2}}\fi%
  \def\@citea{}\@cite{\@for\@citeb:=#2\do%
    {\@citea\def\@citea{; }\@ifundefined
       {b@\@citeb}{{\bf ?}%
       \@warning{Citation `\@citeb' on page \thepage \space undefined}}%
{\csname b@\@citeb\endcsname}}}{#1}}%

%
\def\@biblabel#1{}
\makeatother



\newdimen\bibindent
\bibindent=0.0em
\def\thebibliography#1{\section*{\refname}\list
   {}{\settowidth\labelwidth{[#1]}
   \leftmargin\parindent
   \itemindent -\parindent
   \listparindent \itemindent
   \itemsep 0pt
   \parsep 0pt}
   \def\newblock{}
   \sloppy
   \sfcode`\.=1000\relax}


\setlength{\baselineskip}{12.7pt}

\title{Sensitivity analysis for stopping criteria \\ with application to organ transplantations}

\author{Xingyu Ren\\[12pt]
	Department of Electrical and Computer Engineering\\ 
        Institute for System Research\\
	University of Maryland\\
	College Park, MD 20742, USA\\
\and
Michael C. Fu\\[12pt]
	Robert H. Smith School of Business\\ 
        Institute for System Research\\
	University of Maryland\\
	College Park, MD 20742, USA\\
\and
Steven I. Marcus\\ [12pt]
	Department of Electrical and Computer Engineering\\ 
        Institute for System Research\\
	University of Maryland\\
	College Park, MD 20742, USA\\
}

\maketitle

\section*{ABSTRACT}
We consider a stopping problem and its application to the decision-making process regarding the optimal timing of organ transplantation for individual patients. At each decision period, the patient state is inspected and a decision is made whether to transplant. If the organ is transplanted, the process terminates; otherwise, the process continues until a transplant happens or the patient dies. Under suitable conditions, we show that there exists a control limit optimal policy. We propose a smoothed perturbation analysis (SPA) estimator for the gradient of the total expected discounted reward with respect to the control limit. Moreover, we show that the SPA estimator is asymptotically unbiased.

\section{INTRODUCTION}\label{sec:intro}
This paper is motivated by a kidney transplantation decision-making problem. We consider an end-stage kidney disease (ESKD) patient with a directed living-donor. We assume that the patient is always eligible for transplantation (before they die), and the living-donor organ has a fixed quality and is always available to the patient over the entire decision process. At each decision period, for example, every week or month, the patient health state is inspected and updated, and the decision is whether to transplant depending on the patient health. If the decision is to transplant, the patient receives a \textit{terminal} post-transplantation reward summarizing all the short-term and long-term effect of the transplantation, and the process terminates; otherwise, the patient receives a \textit{intermediate} pre-transplantation reward, and the process continues until a transplantation happens or the patient dies. The goal is to find a policy to maximize the total discounted expected reward. Commonly-used rewards include total discounted expected life years or total discounted quality-adjusted life years (QALYs) \cite{prieto2003problems}.

We propose a Markov decision process (MDP) model \cite{bertsekas2020dynamic} to study this problem, which falls into a special class of MDP models called optimal stopping problems. This type of MDP model has been applied to both liver transplantation \shortcite{alagoz2004optimal,alagoz2007determining,alagoz2007choosing,alagoz2010markov,kaufman2017living,batun2018optimal} and kidney transplantation \shortcite{david1985,bendersky2016deciding,fan2020optimal,ren2022optimal}. Previous work focuses on proving the existence of control limit-type optimal policies. Then, solving MDP problems could be translated into finding an optimal partition of the state space, where each region in the partition is assigned an action; in the scalar-state case, instead of a partition, there is just a single threshold or control limit. Control limit-type policies are important even if they are suboptimal, because they are easy to implement. In general, however, finding an analytic expression for the optimal control limit is difficult. Dynamic programming, one of the most powerful methodologies to solve MDP problems, suffers from the “curse of dimensionality” when the state space or action space is large or even uncountable. In this case, gradient-based optimization methods offer an alternative approach. Moreover, previous work, such as \shortciteN{ren2022optimal}, focuses on the setting of finite state space, while gradient-based methods can be used to solve problems of continuous state spaces. To apply gradient-based optimization methods, one has to compute the gradient of the total expected discounted reward with respect to (w.r.t.) the control limit, where finding an analytic solution is also hard.

In this paper, we focus on estimating the gradient of the total expected discounted reward with respect to the control limit through smoothed perturbation analysis (SPA), a simulation-based method. This is the initial phase of optimization, and fully solving the entire optimization problem will be the focus of future research. The rest of the paper is organized as follows. \Cref{sec:probset} formulates the individual patient organ transplantation problem as a discrete-time, infinite-horizon, \textit{continuous} state space MDP. In \Cref{sec:control}, under suitable conditions, we show the existence of the control limit optimal policy. In \Cref{sec:spa}, we propose an SPA estimator \shortcite{fu2012conditional} for the gradient of the total expected discounted reward w.r.t. the control limit. Moreover, we show that the SPA estimator is asymptotically unbiased. \Cref{sec:num} reports simulation results illustrating the effectiveness of the SPA estimator. The last section offers conclusion and future research directions.

\section{PROBLEM SETTING}\label{sec:probset}
In this section, we introduce components of the MDP model. The set of \textit{decision periods} is the natural numbers $\N=\{0,1,2,\cdots\}$.

Denote the health state of the patient by $h_n\in S_H:=[0,H]$, where a larger value implies worse health. We use an interval $[H_D,H]$, $H_D\in (0,H)$, to represent the death of the patient, i.e., if the patient state $h_n$ is greater than $H_D$, the patient is deceased. (Representing death as an interval rather than a singleton enables the stochastic kernel for the patient state Markov chain to be expressed using density functions, enhancing clarity in deriving structural results.)

Denote the \textit{post-transplantation state} by $P$. The MDP will transition into the absorbing state $P$ if the transplantation happens. The \textit{state space} of the MDP is $\mathcal S=S_H\bigcup\{P\}$, i.e., at each period $n$, the state of the MDP $s_n$ is either a scalar patient state $h_n$, or the post-transplantation state $P$.

Denote the \textit{action} by $a_n$. For each $n\in\N$, $a_n\in \mathcal A=\{W,T\}$ where $\mathcal A$ is the action space including
\begin{itemize}
    \item $W$: \textit{wait} for one more period;
    \item $T$: accept the kidney for \textit{transplantation}.
\end{itemize}
The set of state-action pairs $\sa:=\{(s,a) \ | \ s\in \mathcal S,~ a\in \mathcal A\}$ consists of four mutually disjoint regions, i.e., $\sa=\bigcup_{i=1}^4 K_i$, where $K_1=S_H \times \{W\}$, $K_2=S_H\times \{T\}$, $K_3=(P,W)$, $K_4=(P, T)$.

The \textit{Dynamics} of the MDP is defined as follows: 
\begin{itemize}
    \item If action $W$ is taken, the transplant doesn't happen, and the patient will wait until the next decision period. The patient state evolves according to the Markov transition kernel $\HH(\cdot|\cdot):\B(S_H)\times S_H \mapsto[0,1]$, where $\B(S_H)$ is the collection of Borel subsets of $S_H$. Specifically, given any current patient state $h_n\in S_H$ and any $B\in \B(S_H)$, at the next period, the patient state $h_{n+1}$ will take a value in $B$ with probability (w.p.) $\int_{B} \HH(dh|h_n)$, where $\HH(dh|h_n)$ is a probability measure on measurable space $(S_H,\B(S_H))$.
    \item The transition kernel $\HH$ satisfies the property that once the patient state enters $[H_D,H]$, i.e., the patient dies, they will stay in $[H_D,H]$, and the decision process terminates. In other words, $[H_D,H]$ is an absorbing terminal inverval, i.e., $\int_{H_D}^H\HH(dh|h_n)=1,~\forall h_n\in [H_D,H]$.
    \item If action $T$ is chosen, the state transitions into the absorbing state $P$, the decision process terminates.
\end{itemize}
The general transition kernel $\mathbb{S}:\mathcal{B}(\mathcal{S})\times \sa\mapsto[0,1]$ of the MDP is summarized as follows: for any $n\in \N$, $B\in \mathcal{B}(S_H)$, $h_n\in S_H$,
\begin{align*}
    &\mathbb{S}(s_{n+1}=P \ | \ s_n=P,a_n)=1,\\
    &\mathbb{S}(s_{n+1}\in B  \ | \ s_n=h_n,a_n=W)=\int_{B} \HH(dh|h_n),\\
    &\mathbb{S}(s_{n+1}=P \ | \ s_n=h_n,a_n=T)=1.
\end{align*}

\textit{Reward functions} are defined as follows: given patient state $h_n\in S_H$,
\begin{itemize}
    \item If action $W$ is chosen, an \textit{intermediate pre-transplantation reward} $c(h_n)$ is granted for being alive for one period, where $c(\cdot): S_H \mapsto \R_+$;
    \item If action $T$ is chosen, the patient receives a \textit{terminal post-transplantation reward} $r(h_n)$ that evaluates both the short-term and long-term effect of the transplantation, where $r(\cdot): S_H \mapsto \R_+$. 
\end{itemize}
For $h_n\in[H_D,H]$, i.e., when the patient is deceased, we set $c(h_n)=r(h_n)=0$. The one-stage reward of the MDP $g(\cdot,\cdot):S_H\times \mathcal A \mapsto \R^+$ is given by
\begin{align*}
    g(h,a)=\begin{cases} r(h) & a=T, \\ c(h) &a=W. \end{cases}
\end{align*}

The \textit{objective} is to find a stationary policy $\pi:S_H \mapsto \mathcal A$ maximizing the total discounted expected reward (also known as the value function)
\begin{align*}
    V_\pi(h)=\E(\sum_{k=0}^\infty \lambda^k g(h_k,\pi(h_k)) | h_0=h ),~\forall h\in S_H ,
\end{align*}
where $\lambda\in(0,1)$ is a discount factor. We define the maximum of the total discounted expected reward
\begin{align*}
    V (h)= \max_{\pi\in\Pi} V_\pi(h), ~\forall h\in S_H,
\end{align*}
where $\Pi$ is the set of stationary policies.

\section{CONTROL LIMIT POLICY}\label{sec:control}
In this section, we will show the existence of a control limit optimal policy under suitable conditions, which further expand upon the results of \shortciteN{alagoz2004optimal} to an MDP with a continuous state space. First, we will present several assumptions and preliminary results. 
\begin{assumption} \label{mono-r-con}
Both reward functions $c: S_H \mapsto \R_+$ and $r: S_H \mapsto \R_+$ are continuous and nonincreasing.
\end{assumption}
It follows that for any fixed $a\in \mathcal A$, the one-stage reward $g(s,a)$ is nonincreasing and continuous on $\mathcal S$. Moreover, $g$ is bounded on $\sa$. Because of boundedness of $g$, for any policy $\pi$, its value function $V_\pi(h)$ is bounded on $S_H$.

\begin{definition}(Strong continuity or strong Feller property) 
Let $X,Y$ by Borel spaces. A stochastic kernel $\mathbb T : \mathcal{B}(X)\times Y\mapsto[0,1]$ is said to be strongly continuous if the function $y\mapsto \int v(x)\mathbb T(dx|y)$ belongs to $C_b(Y)$, the set of bounded continuous functions on $Y$, whenever $v\in M_b(X)$, the set of bounded measurable functions on $X$.
\end{definition}

To establish the Bellman optimality condition and value iteration algorithm, we first need to show the strong continuity of the kernel $\mathbb{S}$. We assume the following regularity conditions:

\begin{assumption} \label{dens-h-con} For any $h\in S_H$, the measure $\HH(\cdot|h)$ admits a density $f_\HH(\cdot|h):S_H\mapsto \R_+$, satisfying the following conditions:
\begin{itemize}
    \item For any fixed $h'\in S_H$,  $f_\HH(h'|h)$ is continuous in $h$.
    \item $f_\HH$ is uniformly bounded, i.e., there exists $M>0$ such that for any $(h',h)\in [0,H]^2$, $f_\HH(h'|h)<M$.
\end{itemize}
\end{assumption}

\begin{lemma}
The transition kernel $\mathbb{S}$ is strongly continuous, i.e., for every $u\in M_b(\mathcal 
S)$, $v(s,a)=\int_{S} u(s')\mathbb{S}(ds'|s,a)$ is continuous and bounded on $\sa$.
\end{lemma}
\begin{proof}
    It is enough to show that $\mathbb{S}$ is strongly continuous on $K_1=[0,H]\times \{W\}$, because $\mathbb{S}(\cdot|s,a)$ is a probability mass on some single absorbing state when $(s,a)\in K_i,~i=2,3,4$. Since $a= W$ when $(s,a)\in K_1$, we will drop dependence on $a$ in $v(s,a)$ for simplicity. It suffices to show that $v(h)=\int_S u(s')\mathbb{S}(ds'|s=h,a=W)$ is continuous on $[0,H]$ for any $u\in M_b(\mathcal S)$. We can write
    \begin{align}\label{proofsc1}
        v(h)=\int_0^H u(h') \HH(dh'|h).
    \end{align}
    $v$ is bounded, since the measure $\HH(dh'|h)$ is finite. Then, by \Cref{dens-h-con}, \eqref{proofsc1} can be rewritten as
    \begin{align*}
        \int_0^H u(h')\HH(dh'|h)=\int_0^H u(h')f_\HH(h'|h)dh' .
    \end{align*}
    Take any sequence $\{h_n\}$ such that $h_n\rightarrow h$ as $n\rightarrow \infty$. Since $f_\HH(h'|h)$ is continuous in $h$ for any $h'\in [0,H]$, $u(h')f_\HH(h'|h_n)\rightarrow u(h')f_\HH(h'|h)$ as $n\rightarrow \infty$. Since both $u$ and $f_\HH$ are bounded, by the dominated convergence theorem,
    \begin{align*}
        \int_0^H u(h')f_\HH(h'|h_n)dh' \rightarrow \int_0^H u(h')f_\HH(h'|h)dh'~\text{as } n\rightarrow \infty.
    \end{align*}
    Therefore, $v(h)$ is continuous in $h$.
\end{proof}
Then, by Theorem 4.2.3 and Lemma 4.2.8 in \shortciteN{hernandez2012discrete}, we can establish the Bellman's optimality condition and value iteration algorithm.
\begin{theorem}\label{optimality}
The optimal value function $V(h)$ is the solution of the optimality equation:
\begin{align}\label{bellman-con}
    V(h)=\min\{r(h),c(h)+\lambda \int_0^H V(h') \HH(dh'|h)\},\forall h.
\end{align}
Moreover, the sequence $\{V_k\}$ generated by value iteration
\begin{align*}
\begin{split}
    &V_k(h)=\min\{r(h),c(h)+\lambda \int_0^H V_{k-1}(h') \HH(dh'|h)\},\\
    &V_0(h)=0,\forall h   ,
\end{split}
\end{align*}
is a monotonically nondecreasing sequence, i.e., $V_n(h)\leq V_{n+1}(h),~\forall h\in S_H, k\in\N $ and converges pointwise to $V$, i.e., $V_n \nearrow V^*$.
\end{theorem}

In reliability theory, the increasing failure rate (IFR) property of a probability distribution is a widely-used concept to depict the deterioration of a system \cite{ross1996stochastic}. For a distribution function $F$ with a density or mass function $f$, we say that $F$ is IFR if its failure rate function defined by $f(t)/\bar{F}(t)$, where $\bar{F}:=1-F$ is the complementary (or tail) distribution function, is nondecreasing as a function of $t$. For our purposes, we define the IFR property for the transition kernel of a Markov chain. 

\begin{definition}\label{ifr-con}
Let $\mathbb T$ be a stochastic kernel on $(\mathcal B[0,X],[0,X])$. We say that $\mathbb T$ has the IFR property if for every $x_0\in [0,X]$, $b(x):=\int_{x_0}^X \mathbb T(dx'|x)$ is nondecreasing in $x$.
\end{definition} 

\begin{assumption}\label{ifr-ass}
The transition kernel $\HH$ has the IFR property.
\end{assumption}

\Cref{ifr-ass} has the intuitive explanation that as the patient's health deteriorates, the likelihood of further deterioration increases. The following lemma in \shortciteN{douer1994optimal} provides a necessary and sufficient condition for \Cref{ifr-con}.
\begin{lemma}\label{ifr-equiv}
The stochastic kernel $\mathbb T$ on space $(\mathcal B[0,X],[0,X])$ is IFR if and only if for any bounded, nonnegative and nondecreasing function $v:[0,X]\mapsto \R^+$, $l(x)=\int_0^X v(x)\mathbb T(dx'|x)$ is also nondecreasing.
\end{lemma}

The monotinicity of the value function $V$ can be easily shown from \eqref{bellman-con} in \Cref{optimality} and \Cref{ifr-equiv}.
\begin{theorem}\label{th:mono}
Under \Cref{mono-r-con,dens-h-con,ifr-ass}, the value function $V$ is nonincreasing on $S_H$.
\end{theorem}
\Cref{th:mono} implies that the patient's overall benefit, e.g., the total QALYs, will not increase if the patient health deteriorates.

\Cref{controllimit} provides sufficient conditions for the existence of a control limit optimal policy. Specifically, \Cref{controllimit} shows that there exists an optimal policy $\pi^*$ that partitions $S_H$ into two intervals:
\begin{align}\label{control}
\begin{split}
    \pi^*(h)=\begin{cases}
    W &\text{if } h<\theta^*,\\
    T &\text{if } h\geq \theta^*,
    \end{cases}   
\end{split}
\end{align} 
where $\theta^*$ is called the optimal \textit{control limit} (or threshold). The optimal action to take depends only on whether the state $h$ is greater than or less than the control limit $\theta^*$, and solving the MDP problem boils down to finding this optimal threshold. To prove \Cref{controllimit}, we need the following lemma \shortcite{alagoz2007determining} and several additional assumptions. 
\begin{lemma}\label{lemma:schaefer}
Let $v$ be a bounded, nonnegative, and nonincreasing function. If the stochastic kernel $\mathbb T$ defined on $(\mathcal B[0,X],[0,X])$ is IFR and admits a uniformly bounded density function $f:[0,X]^2\mapsto \R_+$, then the following results hold: for any $x_1<x_2$,
\begin{itemize}
    \item $\int_0^{x_1}v(x)(f(x|x_1)-f(x|x_2))dx\geq v(x_1)\int_0^{x_1}(f(x|x_1)-f(x|x_2))dx$;
    \item $\int_{x_1}^X v(x)(f(x|x_1)-f(x|x_2))dx\geq v(x_1)\int_{x_1}^X(f(x|x_1)-f(x|x_2))dx$.
\end{itemize}
\end{lemma}
\begin{proof}
We only provide a proof for the first part, as the second part can be proved in a similar way. First, we consider the case that $v$ is simple function, i.e., $v(x)=\sum_{i=1}^n v_i \1_{A_i}(x)$ for some $n$, where $\1_{A_i}$ is the indicator function of set $A_i$ and $\{A_i\}_{i=1,\cdots,n}$ is a partition of $[0,x_1]$. $v$ nonincreasing allows us to take each $A_i$ to be an interval. Assume that $v_1\geq v_2\geq \cdots \geq v_n$. Then,
\begingroup
\allowdisplaybreaks
\begin{align*}
    \int_0^{x_1}v(x)(f(x|x_1)-f(x|x_2))dx&=\sum_{i=1}^n v_i \int_{A_i} (f(x|x_1)-f(x|x_2))dx\\
    &=v_1 \int_{A_1} (f(x|x_1)-f(x|x_2))dx+\sum_{i=2}^n v_i \int_{A_i} (f(x|x_1)-f(x|x_2))dx\\
    &\geq v_2 \int_{A_1} (f(x|x_1)-f(x|x_2))dx+\sum_{i=2}^n v_i \int_{A_i} (f(x|x_1)-f(x|x_2))dx\\
    &\geq v_n \int_0^{x_1}(f(x|x_1)-f(x|x_2))dx \\
    &= v(x_1)\int_0^{x_1}(f(x|x_1)-f(x|x_2))dx.
\end{align*}
\endgroup
For any general nonincreasing function $v$, we can take a monotone sequence of simple functions $\{v_n\}$ such that $v_n\nearrow v$ pointwise. Since $f$ is uniformly bounded, the result follows from the dominated convergence theorem.
\end{proof}

\begin{assumption}\label{h-based-con1}
For any $h_1<h_2 \text{ and } h_0$,
    \begin{align*}
        \int_{h_0}^{H_D} f_\HH(h|h_1) dh \leq \int_{h_0}^{H_D} f_\HH(h|h_2) dh.
    \end{align*}
\end{assumption}

\Cref{h-based-con1} takes a similar form as the IFR property and can be interpreted similarly, but \Cref{h-based-con1} is neither a sufficient nor a necessary condition for the IFR property of $\HH$.

\begin{assumption}\label{h-based-con2}
For any $h_1<h_2$,
    \begin{align*}\label{rate-condition-2-con}
    \frac{r(h_1)-r(h_2)}{r(h_2)} \leq  \lambda\left(\int_{H_D}^H f_\HH(h|h_2) dh - \int_{H_D}^H f_\HH(h|h_1) dh \right).
    \end{align*}
\end{assumption}

\Cref{h-based-con2} has an intuitive explanation that, as the patient health becomes worse, the increment of the probability of death during waiting is greater than the \textit{marginal} reduction in the transplantation reward. \shortciteN{alagoz2004optimal} presents empirical evidence that \Cref{h-based-con1,h-based-con2} are applicable in the context of living-donor liver transplantation.

\begin{theorem}\label{controllimit}
Under \Cref{mono-r-con,dens-h-con,ifr-ass,h-based-con1,h-based-con2}, there exists a control limit optimal policy taking the form of \eqref{control}. 
\end{theorem}

\begin{proof}
By contradiction, suppose that for some $h_1<h_2$, $T\in a^*(h_1)$, but $a^*(h_2)=W$, where $a^*(h)$ is the set of optimal actions at $h$. Then, we have
\begin{align*}
    r(h_1)&\geq c(h_1)+\lambda \int_0^{H_D}  V(h') f_\HH(h'|h_1)dh',\\
    r(h_2)&<c(h_2)+\lambda \int_0^{H_D} V(h') f_\HH(h'|h_2)dh'.
\end{align*}
Then,
\begin{align*}
    r(h_1)-r(h_2)&> c(h_1)-c(h_2)+\lambda \int_0^{H_D}  V(h') (f_\HH(h'|h_1)-f_\HH(h'|h_2))dh'\\
    &\geq \lambda \int_0^{h_1}  V(h') (f_\HH(h'|h_1)-f_\HH(h'|h_2))dh' +\lambda\int_{h_1}^{H_D}  V(h') (f_\HH(h'|h_1)-f_\HH(h'|h_2))dh'\\
    &\geq \lambda V(h_1) \int_0^{H_D}  (f_\HH(h'|h_1)-f_\HH(h'|h_2))dh'\\
    &=\lambda V(h_1) \left(\int_{H_D}^H f_\HH(h|h_2) dh - \int_{H_D}^H f_\HH(h|h_1) dh \right),
\end{align*}
where the second inequality follows from \Cref{mono-r-con}, and the third inequality follows from \Cref{lemma:schaefer}. By \Cref{h-based-con2}, we have $V(h_1)<r(h_2)$. Since $a^*(h_2)=W$, \begin{align*}
    r(h_2)<V(h_2)\leq V(h_1)\leq V(h_1),
\end{align*}
which is a contradiction. Therefore, for any $h_1<h_2$, $T\in a^*(h_1)$ implies that $T\in a^*(h_2)$.
\end{proof}
\Cref{controllimit} has an intuitive explanation that the patient should be transplanted if and only if their health status is worse than some threshold (recall that a larger patient state $h_n$ implies the worse health status).

\section{Smoothed perturbation analysis (SPA) estimator}\label{sec:spa}
In this section, we propose an SPA estimator for the gradient of the value function w.r.t. the control limit. Throughout this section, we suppose that a control limit policy with control limit $\theta$, denoted by $\pi_\theta$, is implemented. For a fixed initial condition $h_0\in S_H$, let $V(\theta)$ be the value function associated with $\pi_\theta$ and assume that $V(\theta)$ is differentiable w.r.t. $\theta$. We want to estimate the gradient of $V(\theta)$ w.r.t. $\theta$. 

We denote by $h_k(\theta),~\forall k$ the patient state at period $k$, under control limit policy $\pi_\theta$. For fixed $n\in\N$ and $h_0\in S_H$, we consider the following sample performance
\begin{align*}
    v_n(\theta) = \sum_{k=0}^n \lambda^k g(h_k,\pi_\theta(h_k)),
\end{align*}
i.e., the total discounted reward until period $n$ under policy $\pi_\theta$. Notice that
\begin{align*}
    g(h,\pi_\theta(h)) = \begin{cases} r(h) &h\geq \theta, \\ c(h) &h<\theta. \end{cases}
\end{align*}
Given a nominal sample path under policy $\pi_\theta$, suppose that a perturbation of $\Delta\theta$ is introduced to construct a sample path under policy $\pi_{\theta+\Delta\theta}$, called the perturbed sample path. An infinitesimal perturbation analysis (IPA) estimator comes from taking the derivative of each $g(h_k,\pi_\theta(h_k))$ while assuming that the event $\{h_n\neq P,~h_n<\theta\}$ is unchanged, i.e., the perturbation $\Delta\theta$ results in no change in the transplant decision (thus there is also no change in the sample path). Under this assumption,
\begin{align*}
    \frac{dh_k(\theta)}{d\theta}=0~\text{w.p. }1,
\end{align*}
and
\begin{align*}
    \frac{dg(h_k,\pi_\theta(h_k))}{d\theta} &= \frac{\partial g(h_k,\pi_\theta(h_k))}{\partial \theta} + \frac{\partial g(h_k,\pi_\theta(h_k))}{\partial h}\frac{dh_k(\theta)}{d\theta}\\
    &=\frac{\partial g(h_k,\pi_\theta(h_k))}{\partial \theta}\\
    &= 0~\text{w.p. }1,
\end{align*}
because for fixed $h$, $g(h,\pi_\theta(h))$ is a function of $\theta$ with only one discontinuity of zero measure.

However, the IPA estimator does not capture the discrete changes that occur, for example, when the action in some period changes from ``transplant" in the nominal sample path to ``wait" in the perturbed sample path. We use SPA to calculate discrete changes caused by the change of the action. Specifically, by conditioning on suitable quantities, we compute the conditional expectation on the change in $v_n$, and take $\Delta\theta\rightarrow 0$. In this MDP model, discrete changes may potentially occur only at $M(n)=\min\{i\leq n : h_i\geq \theta\} $, i.e., the period when a transplantation happens in the nominal sample path. Conditioned on $M(n)$, the event $\{h_{M(n)}\geq\theta\}$ is equivalent to $\{\alpha_n \geq 0\}$, where $\alpha_n:=h_{M(n)}-\theta$. The action in the perturbed sample path alters if $\alpha_n < \Delta\theta$. Conditioned on $h_{M(n)-1}$, the SPA estimator is expressed as
\begin{align*}
    \left(\frac{\partial v_n(\theta)}{\partial \theta}\right)_{SPA} &=\lim_{\Delta\theta\rightarrow 0} \E(\Delta v_n(\theta) \1\{\alpha_n\leq \Delta\theta\}|h_{M(n)-1}) /\Delta\theta\\
    &=\lim_{\Delta\theta\rightarrow 0} \E(\Delta v_n(\theta)|\alpha_n\leq \Delta\theta|h_{M(n)-1})\P(\alpha_n\leq \Delta\theta|h_{M(n)-1})/\Delta\theta,
\end{align*}
where
\begingroup
\allowdisplaybreaks
\begin{align*}
    \P(\alpha_n\leq \Delta\theta|h_{M(n)-1}) & = \P(\alpha_n\leq \Delta\theta | \alpha_n\geq 0,h_{M(n)-1})\\
    &= \P(  \theta\leq h_{M(n)}\leq \theta+ \Delta\theta|h_{M(n)-1})/ \P(h_{M(n)}\geq \theta|h_{M(n)-1})\\
    &=\frac{\int_\theta^{\theta+\Delta\theta} \HH (dh|h_{M(n)-1})}{\int_\theta^{H} \HH (dh|h_{M(n)-1})}   .
\end{align*}
\endgroup
Therefore,
\begin{align*}
    \lim_{\Delta\theta\rightarrow 0} \P(\alpha_n\leq \Delta\theta|h_{M(n)-1})/\Delta\theta = \frac{f_\HH(\theta|h_{M(n)-1})}{\int_\theta^{H} \HH(dh|h_{M(n)-1})}.
\end{align*}
The term $\lim_{\Delta\theta\rightarrow 0} \E(\Delta v_n(\theta)|\alpha_n\leq \Delta\theta,h_{M(n)-1})$ is the extra accumulated reward caused by the change of the action, which is given by
\begin{align*}
    \lim_{\Delta\theta\rightarrow 0} \E(\Delta v_n(\theta)|\alpha_n\leq \Delta\theta,h_{M(n)-1}) = \lambda^{M(n)}(r(\theta)-c(\theta)) - \E\left(\sum_{i=M(n)+1}^{n}\lambda^i g(h_i,\pi_\theta(h_i))|h_{M(n)}=\theta^-\right).
\end{align*}
The final SPA gradient estimator is given by
\begin{align*}
    \left(\frac{\partial v_n(\theta)}{\partial \theta}\right)_{SPA}=\frac{f_\HH(\theta|h_{M(n)-1})}{\int_\theta^{H} \HH (dh|h_{M(n)-1})} \left(\lambda^{M(n)}(r(\theta)-c(\theta)) - \E\left(\sum_{i=M(n)+1}^{n}\lambda^i g(h_i,\pi_\theta(h_i))|h_{M(n)}=\theta^-\right)\right).
\end{align*}
Now we formally show that $\left(\frac{\partial v_n(\theta)}{\partial \theta}\right)_{SPA}$ is an asymptotically unbiased estimator, i.e.,
\begin{align*}
    \E\left(\frac{\partial v_n(\theta)}{\partial \theta}\right)_{SPA} \rightarrow \frac{\partial V(\theta)}{\partial\theta} \text{ as }n\rightarrow\infty.
\end{align*}
First, we show that $\left(\frac{\partial v_n(\theta)}{\partial \theta}\right)_{SPA}$ is an unbiased estimator of $\frac{\partial \E (v_n(\theta))}{\partial\theta}$.
\begin{theorem}\label{unbias}
Under \Cref{mono-r-con,dens-h-con}, $\left(\frac{\partial v_n(\theta)}{\partial \theta}\right)_{SPA}$ is an unbiased estimator of $\frac{\partial \E (v_n(\theta))}{\partial\theta}$, i.e., 
\begin{align*}
    \E\left(\frac{\partial v_n(\theta)}{\partial \theta}\right)_{SPA} = \frac{\partial \E (v_n(\theta))}{\partial\theta}.
\end{align*}
\end{theorem}
\begin{proof}
We define the following events: for fixed $\Delta \theta$,
\begin{align*}
    &A_k=\{h_i(\theta)< \theta \text{ or } h_i(\theta)\geq \theta+\Delta\theta,~i=1,\cdots ,k\},~\forall k,\\
    &B_k=A_k^\sim,
\end{align*}
i.e., $A_k$ is the event that a perturbation of size $\Delta\theta$ doesn't cause a change in the transplant decision until time $k$. Then, we can write
\begin{align*}
    \frac{\partial \E (v_n(\theta))}{\partial\theta}&=\lim_{\Delta \theta\rightarrow 0} \left( \frac{\E((v_n(\theta+\Delta\theta)-v_n(\theta))\1\{A_n\})}{\Delta\theta} + \frac{\E((v_n(\theta+\Delta\theta)-v_n(\theta))\1\{B_n\})}{\Delta\theta} \right),
\end{align*}
where the first term is zero, because the perturbed sample path and the nominal sample path are the same, conditioned on the event $A_n$. We write the term $\E((v_n(\theta+\Delta\theta)-v_n(\theta))\1\{B_n\})$ as
\begin{align*}
    \E(\E((v_n(\theta+\Delta\theta)-v_n(\theta))\1\{B_n\}|h_{M(n-1)})).
\end{align*}
Note that the event $B_n$ is equivalent to the event $\{M(n) \text{ is non-empty},~h_{M(n)}(\theta)\geq  \theta,~h_{M(n)}(\theta)<\theta+\Delta\theta\}$. Then,
\begin{align*}
    \E(v_n(\theta+\Delta\theta)\1\{B_n\}|h_{M(n)-1})&=\E(v_n(\theta+\Delta\theta)|h_{M(n)-1},\1\{B_n\})\P(\1\{B_n\}|h_{M(n)-1})\\
    &= \E(v_n(\theta+\Delta\theta)|h_{M(n)-1},\1\{B_n\})\\  &\times \P(\theta \leq h_{M(n)}<\theta+ \Delta\theta | h_{M(n)-1})\\
    &\leq \frac{\max_h\{c(h),r(h)\}}{1-\lambda} \int_\theta^{\theta+\Delta\theta} f_\HH(h|h_{M(n-1)})dh\\
    &\leq \frac{M\Delta\theta\max_h\{c(h),r(h)\}}{1-\lambda},
\end{align*}
where the first inequality follows from the fact that $v_n(\theta)\leq \sum_{i=0}^\infty\lambda^i \max_h\{c(h),r(h)\}$, and the second inequality follows from \Cref{dens-h-con} that $f_\HH$ is uniformly bounded by $M$. Therefore, $\E(v_n(\theta+\Delta\theta)\1\{B_n\}|h_{M(n-1)})/\Delta\theta$ is uniformly bounded for any $\Delta\theta$. By the dominated convergence theorem,
\begin{align*}
    \lim_{\Delta\theta\rightarrow 0} \frac{\E(v_n(\theta+\Delta\theta)\1\{B_n\})}{\Delta\theta}&=\lim_{\Delta\theta\rightarrow 0} \E\left(\frac{\E(v_n(\theta+\Delta\theta)\1\{B_n\}|h_{M(n)-1})}{\Delta\theta}\right)\\
    &=\E\left(\lim_{\Delta\theta\rightarrow 0} \frac{\E(v_n(\theta+\Delta\theta)\1\{B_n\}|h_{M(n)-1})}{\Delta\theta}\right)\\
    &=\E\left( \lim_{\Delta\theta\rightarrow 0} \frac{\P(\1\{B_n\}|h_{M(n)-1})}{\Delta\theta}\times \lim_{\Delta\theta\rightarrow 0} \E(v_n(\theta+\Delta\theta)|\1\{B_n\},h_{M(n)-1})\right)\\
    &=\E\Biggl( \frac{f_\HH(\theta|h_{M(n)-1})}{\int_\theta^{H} \HH (dh|h_{M(n)-1})}\\ &\times \left(\sum_{i=0}^{M(n)-1}\lambda^i c(h_i)+\lambda^{M(n)}c(\theta) + \E\left(\sum_{i=M(n)+1}^{n}\lambda^i g(h_i,\pi_{\theta}(h_i))|h_{M(n)}=\theta^-\right)\right)\Biggr).
\end{align*}
Similarly, we can derive
\begin{align*}
    \lim_{\Delta\theta\rightarrow 0} \frac{\E(v_n(\theta)\1\{B_n\})}{\Delta\theta} = \E\left( \frac{f_\HH(\theta|h_{M(n)-1})}{\int_\theta^{H} \HH (dh|h_{M(n)-1})}\times \left(\sum_{i=0}^{M(n)-1}\lambda^i c(h_i) +\lambda^{M(n)}r(\theta)\right)\right).
\end{align*}
It follows that
\begin{align*}
    \lim_{\Delta\theta\rightarrow 0} \frac{\E((v_n(\theta+\Delta\theta)-v_n(\theta))\1\{B_n\})}{\Delta\theta}&=\E\Biggl( \frac{f_\HH(\theta|h_{M(n)-1})}{\int_\theta^{H} \HH (dh|h_{M(n)-1})}\\ &\times (\lambda^{M(n)}(c(\theta)-r(\theta)) + \E\left(\sum_{i=M(n)+1}^{n}\lambda^i g(h_i,\pi_{\theta}(h_i))|h_{M(n)}=\theta^-\right)\Biggr).
\end{align*}
Therefore, we conclude that $\E\left(\frac{\partial v_n(\theta)}{\partial \theta}\right)_{SPA} = \frac{\partial \E (v_n(\theta))}{\partial\theta}$.
\end{proof}

\begin{theorem}
Under \Cref{mono-r-con,dens-h-con}, $\left(\frac{\partial v_n(\theta)}{\partial \theta}\right)_{SPA}$ is an asymptotically unbiased estimator, i.e.,
\begin{align*}
    \lim_{n\rightarrow\infty} \E\left(\frac{\partial v_n(\theta)}{\partial \theta}\right)_{SPA} = \frac{\partial V(\theta)}{\partial\theta}.
\end{align*}
\end{theorem}

\begin{proof}
Note that $\lim_{n\rightarrow\infty} \E(v_n(\theta)) = \E (\lim_{n\rightarrow\infty}v_n(\theta))=V(\theta)$ because the sequence of random variables $\{v_n\}_{n\in\N}$ is uniformly bounded and converges pointwise. Since we already proved $\E\left(\frac{\partial v_n(\theta)}{\partial \theta}\right)_{SPA} = \frac{\partial \E (v_n(\theta))}{\partial\theta}$ in \Cref{unbias}, it remains to show that
\vspace*{-18pt}
$$
~~~~~~~~~~~~~~~~~~~~~~
    \lim_{n\rightarrow\infty} \frac{\partial \E (v_n(\theta))}{\partial\theta} = \frac{\partial \lim_{n\rightarrow\infty} \E (v_n(\theta))}{\partial\theta}.
$$
i.e., passing a derivative through a limit. By Theorem 8.2.3 in \shortciteN{bartle2000introduction}, it suffices to show that $\left(\frac{\partial \E (v_n(\theta))}{\partial\theta}\right)_{n\in\N}$ is a uniformly convergent sequence on $S_H$. Note that
\begin{align*}
    \biggl|\frac{\partial \E v_n(\theta)}{\partial\theta}-\frac{\partial \E v_{n-1}(\theta)}{\partial\theta}\biggr| &= \lambda^{n} \biggl|\frac{\partial \E g(h_{n},\pi_\theta(h_{n}))}{\partial\theta}\biggr|. 
\end{align*}
It suffices to show that $\biggl|\frac{\partial \E g(h_{n},\pi_\theta(h_{n}))}{\partial\theta}\biggr|$ is bounded, which can be shown similarly as in \Cref{unbias}. We write
\begin{align*}
    \frac{\partial \E g(h_{n},\pi_\theta(h_{n}))}{\partial\theta}&=\lim_{\Delta\theta\rightarrow 0}\frac{\E g(h_{n}(\theta+\Delta\theta),\pi_\theta(h_{n}(\theta+\Delta\theta))) - \E g(h_{n}(\theta),\pi_\theta(h_{n}(\theta)))}{\Delta\theta}.
\end{align*}
Similar to the proof of \Cref{unbias}, we can write $\E g(h_{n}(\theta),\pi_\theta(h_{n}(\theta)))$ as
\begin{align*}
    \E g(h_{n}(\theta),\pi_\theta(h_{n}(\theta))) &= \E (\E(g(h_{n}(\theta),\pi_\theta(h_{n}(\theta)))\1\{B_n\}|h_{M(n)-1})), \mbox{~where}  \\
    \E(g(h_{n}(\theta),\pi_\theta(h_{n}(\theta)))\1\{B_n\}|h_{M(n)-1})&=\P(1\{B_n\}|h_{M(n)-1}) \E(g(h_{n}(\theta),\pi_\theta(h_{n}(\theta)))|\1\{B_n\},h_{M(n)-1})\\
    &\leq \max_h\{c(h),r(h)\} \int_\theta^{\theta+\Delta\theta} f_\HH(h|h_{M(n)-1})dh\\
    &\leq \Delta\theta M \max_h\{c(h),r(h)\}.
\end{align*}
Similarly, $\E g(h_{n}(\theta+\Delta\theta),\pi_\theta(h_{n}(\theta+\Delta\theta)))\leq \Delta\theta M \max_h\{c(h),r(h)\}$. It follows that
$$
\biggl|\frac{\partial \E g(h_{n},\pi_\theta(h_{n}))}{\partial\theta}\biggr|
\leq \hspace*{-3pt} \lim_{\Delta\theta\rightarrow 0} \hspace*{-3pt} 
\frac{\E g(h_{n}(\theta+\Delta\theta),\pi_\theta(h_{n}(\theta+\Delta\theta))) + \E g(h_{n}(\theta),\pi_\theta(h_{n}(\theta)))}{\Delta\theta}
\hspace*{-3pt} =2M \max_h\{c(h),r(h)\},
$$
which is independent of $\theta$. Therefore, $\biggl|\frac{\partial \E v_n(\theta)}{\partial\theta}-\frac{\partial \E v_{n-1}(\theta)}{\partial\theta}\biggr|$ converges uniformly to zero, and $\left(\frac{\partial \E (v_n(\theta))}{\partial\theta}\right)_{n\in\N}$ is a uniformly convergent sequence on $S_H$. Thus we can pass the derivative through the limit.
\end{proof}

\section{Simulation example}\label{sec:num}
In this section, we present a simple simulation example to demonstrate the performance of the SPA estimator. We consider an MDP where the decision period is half a year, and the patient state $h_n$ takes values in $S_H=[0,1]$. Suppose that the control limit policy $\pi_\theta$ is implemented for some $\theta\in(0,1)$. For any $n$, 
conditioned on $h_n\in[0,\theta]$, $h_{n+1}$ is uniformly distributed over $[h_n,1]$, i.e., $f_\HH(h'|h)=\1\{1\geq h'\geq h\}/(1-h),~\forall h,h'\in[0,1]$. Therefore, the patient health never improves, and it is straightforward to check that $f_\HH$ satisfies the IFR property. The rewards are defined in terms of expected life years. We define the intermediate pre-transplantation reward $c(h)\equiv0.5$ (years), and the terminal post-transplantation reward function $r(h)=8(1-h)$ (years). Then, the SPA estimator is given by
\begin{align*}
    \left(\frac{\partial v_n(\theta)}{\partial \theta}\right)_{SPA}&=\frac{f_\HH(\theta|h_{M(n)-1})}{\int_\theta^{H} \HH (dh|h_{M(n)-1})} \left(\lambda^{M(n)}(r(\theta)-c(\theta)) - \E\left(\sum_{i=M(n)+1}^{n}\lambda^i g(h_i,\pi_\theta(h_i))|h_{M(n)}=\theta^-\right)\right)\\
    &=\frac{1}{1-\theta} (\lambda^{M(n)}(r(\theta)-c(\theta)) - \E(\lambda^{M(n)+1} r(h_{M(n)+1})|h_{M(n)}=\theta^-))\\
    &=\frac{1}{1-\theta} (\lambda^{M(n)}(8-8\theta-0.5) - \E(\lambda^{M(n)+1} 8(1-h_{M(n)+1})|h_{M(n)}=\theta^-)),
\end{align*}
which we compare with the symmetric finite difference (FD) estimator
\vspace*{-4pt}
$$
    \left(\frac{\partial v_n(\theta)}{\partial \theta}\right)_{FD}=\frac{v_n(\theta+\frac{\delta}{2})-v_n(\theta-\frac{\delta}{2})}{\delta},
$$
where $\delta$ is the size of the symmetric difference. We compute both derivative estimators at $\theta=0.2,0.5,0.8$ and test with $\delta=0.01,0.05,0.1$ and number of replications $N = 10^2,10^4,10^6$. Simulation results are shown in \Cref{resl} and \Cref{figure}. We have the following observations:
\vspace*{-5pt}
\begin{itemize}
    \item For a small number of replications, SPA has much smaller bias and standard error (SE) than FD.
    \item FD at $\delta=0.1$ has a large bias that can be reduced at the expense of variance.
    \item FD at $\delta=0.01$ is almost unbiased but has much larger variance than SPA.
    \item For a fixed number of replications, the standard error of the SPA estimator is almost identical at different $\theta$, whereas the precision of the FD estimator is proportional to the derivative.
\end{itemize}
\vspace*{-14pt}

\begin{table}[t]
\setlength{\belowcaptionskip}{-.01pt}
\setlength\extrarowheight{-.01pt}
\centering
\caption{\baselineskip8pt Simulation results for sensitivity of value function $V(\theta)$ w.r.t. $\theta$ (standard errors in parentheses).}
\label{resl}
\begin{tabular}{c|c|c|c|c|c}
$N$ & $\theta$ & SPA & FD($\delta=0.01$) & FD($\delta=0.05$) & FD($\delta=0.1$) \\ \hline
                   & $0.2$ &$-3.199(0.242)$     &$-8.251(5.394)$    &$-4.786(1.811)$    &$-2.065(0.871)$    \\ \cline{2-6} 
$10^2$                & $0.5$ &$-2.668(0.233)$     &$-5.523(3.182)$   &$-2.560(1.076)$    &$-3.512(0.780)$    \\ \cline{2-6} 
                   & $0.8$ &$-1.313(0.225)$     &$-3.083(1.411)$    &$-1.326(0.552)$    &$-0.205(0.242)$    \\ \hline
                   & $0.2$ &$-3.371(0.023)$     &$-3.281(0.349)$    &$-3.503(0.155)$    &$-3.253(0.104)$    \\ \cline{2-6} 
$10^4$               & $0.5$ &$-2.997(0.023)$     &$-3.120(0.265)$    &$-2.991(0.109)$    &$-2.920(0.071)$    \\ \cline{2-6} 
                   & $0.8$ &$-1.515(0.022)$     &$-1.306(0.114)$    &$-1.144(0.043)$    &$-0.527(0.028)$    \\ \hline
                   & $0.2$ &$-3.403(0.002)$     &$-3.446(0.036)$    &$-3.346(0.016)$    &$-3.310(0.010)$    \\ \cline{2-6} 
$10^6$              & $0.5$ &$-3.019(0.002)$     &$-2.948(0.026)$    &$-2.945(0.011)$    &$-2.867(0.007)$    \\ \cline{2-6} 
                   & $0.8$ &$-1.517(0.002)$     &$-1.413(0.011)$    &$-1.106(0.004)$    &$-0.527(0.003)$   
\end{tabular}
\end{table}

\begin{figure}[b]
\captionsetup[subfigure]{aboveskip=-.2pt,belowskip=-.2pt}
     \centering
     \begin{subfigure}[b]{0.36\textwidth}
         \centering
         \includegraphics[width=\textwidth]{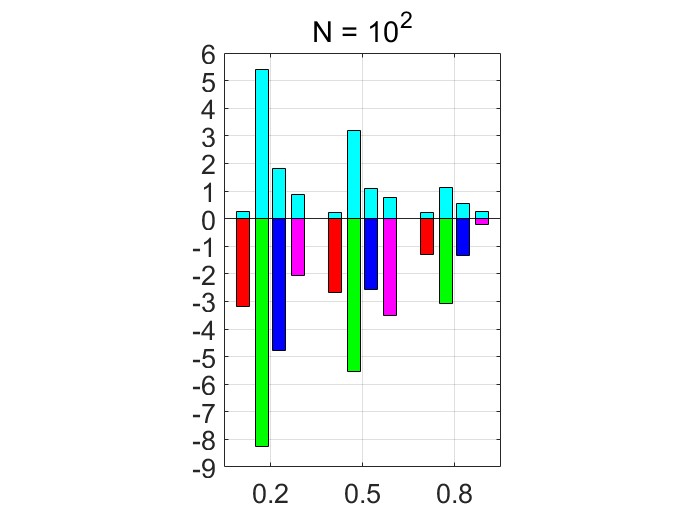}
     \end{subfigure}
     \hspace*{-4.3em}
     \begin{subfigure}[b]{0.36\textwidth}
         \centering
         \includegraphics[width=\textwidth]{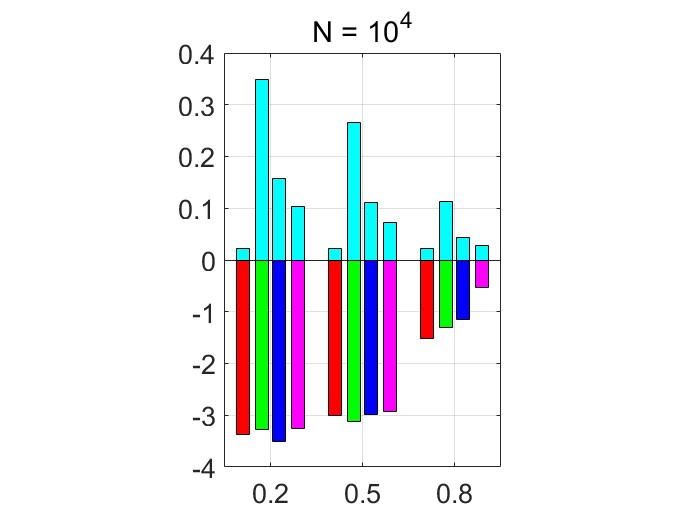}
     \end{subfigure}
     \hspace*{-2.5em}
     \begin{subfigure}[b]{0.36\textwidth}
         \centering
         \includegraphics[width=\textwidth]{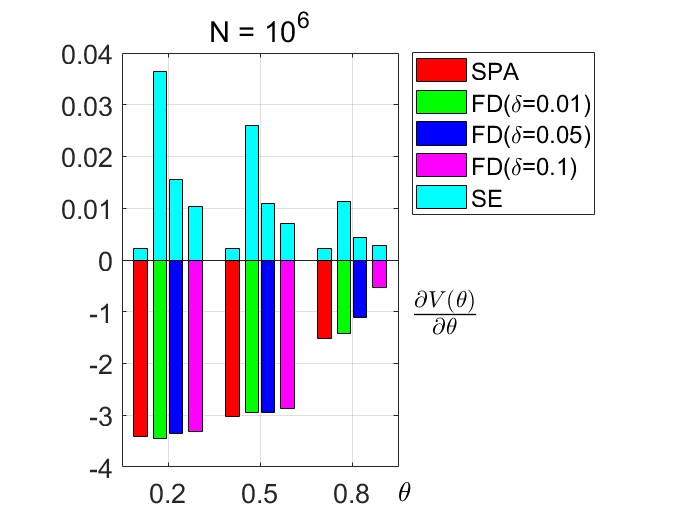}
     \end{subfigure}
        \caption{\baselineskip8pt Simulation results for the sensitivity of the value function $V(\theta)$ and their standard errors (SEs).} \label{figure}
\end{figure}

\section{Summary and future research}\label{sec:conc}
We proposed a continuous-state MDP model to study the optimal timing of organ transplantation. Under suitable conditions, we proved that there exists a control limit optimal policy. We derived an SPA estimator for the gradient of the value function w.r.t the control limit, which is useful in computing the optimal control limit by gradient-based simulation optimization. Furthermore, we proved that the SPA estimator is asymptotically unbiased and demonstrated its effectiveness using a simulation example. Solving for the optimal control limit through gradient-based optimization methods will be the focus of future research. Another future research direction is to consider the situation where the donor organ's quality and availability may vary over time. Finally, because implementing the SPA estimator requires additional simulation beyond the nominal sample path, comparing it with other unbiased gradient estimators, for example, the generalized likelihood ratio (GLR) method proposed in \shortciteN{peng2018new} and the measure-valued differentiation (MVD) method in \shortciteN{heidergott2023gradient}, is an important topic warranting further investigation.

\appendix

\footnotesize

\bibliographystyle{wsc}

\bibliography{demobib}

@article{heidergott2023gradient,
  title={Gradient Estimation for Smooth Stopping Criteria},
  author={Heidergott, Bernd and Peng, Yijie},
  journal={Advances in Applied Probability},
  volume={55},
  number={1},
  pages={29--55},
  year={2023},
  publisher={Cambridge University Press}
}

@book{bartle2000introduction,
  title={Introduction to Real Analysis},
  author={Bartle, Robert G and Sherbert, Donald R},
  year={2010},
  publisher={John Wiley \& Sons},
  edition={4th},
  address = {New York}
}

@article{ren2022optimal,
  title={Optimal Acceptance of Incompatible Kidneys},
  author={Ren, Xingyu and Fu, Michael C and Marcus, Steven I},
  journal={arXiv preprint arXiv:2212.01808},
  year={2022}
}

@article{prieto2003problems,
  title={Problems and Solutions in Calculating Quality-Adjusted Life Years (QALYs)},
  author={Prieto, Luis and Sacrist{\'a}n, Jos{\'e} A},
  journal={Health and Quality of Life Outcomes},
  volume={1},
  pages={1--8},
  year={2003},
  publisher={Springer}
}

@book{ross1996stochastic,
  title={Stochastic Processes},
  author={Ross, Sheldon M},
  year={1996},
  publisher={John Wiley \& Sons},
  address = {New York},
  edition = {2nd}
}

@article{peng2018new,
  title={A New Unbiased Stochastic Derivative Estimator for Discontinuous Sample Performances With Structural Parameters},
  author={Peng, Yijie and Fu, Michael C and Hu, Jian-Qiang and Heidergott, Bernd},
  journal={Operations Research},
  volume={66},
  number={2},
  pages={487--499},
  year={2018},
  publisher={INFORMS}
}

@book{fu2012conditional,
  title={Conditional Monte Carlo: Gradient Estimation and Optimization Applications},
  author={Fu, Michael C and Hu, Jian-Qiang},
  year={1997},
  publisher={Kluwer Academic},
  address = {Boston}
}

@article{david1985,
  title={A Time-Dependent Stopping Problem With Application to Live Organ Transplants},
  author={David, Israel and Yechiali, Uri},
  journal={Operations Research},
  volume={33},
  number={3},
  pages={491--504},
  year={1985},
  publisher={INFORMS}
}

@article{alagoz2004optimal,
  title={The Optimal Timing of Living-Donor Liver Transplantation},
  author={Alagoz, Oguzhan and Maillart, Lisa M and Schaefer, Andrew J and Roberts, Mark S},
  journal={Management Science},
  volume={50},
  number={10},
  pages={1420--1430},
  year={2004},
  publisher={INFORMS}
}

@article{alagoz2010markov,
  title={Markov Decision Processes: {A} Tool for Sequential Decision Making Under Uncertainty},
  author={Alagoz, Oguzhan and Hsu, Heather and Schaefer, Andrew J and Roberts, Mark S},
  journal={Medical Decision Making},
  volume={30},
  number={4},
  pages={474--483},
  year={2010},
  publisher={SAGE Publications Sage CA: Los Angeles, CA}
}

@article{alagoz2007determining,
  title={Determining the Acceptance of Cadaveric Livers Using an Implicit Model of the Waiting List},
  author={Alagoz, Oguzhan and Maillart, Lisa M and Schaefer, Andrew J and Roberts, Mark S},
  journal={Operations Research},
  volume={55},
  number={1},
  pages={24--36},
  year={2007},
  publisher={INFORMS}
}

@article{batun2018optimal,
  title={Optimal Liver Acceptance for Risk-Sensitive Patients},
  author={Batun, Sakine and Schaefer, Andrew J and Bhandari, Atul and Roberts, Mark S},
  journal={Service Science},
  volume={10},
  number={3},
  pages={320--333},
  year={2018},
  publisher={INFORMS}
}

@book{bertsekas2020dynamic,
  title={Dynamic Programming and Optimal Control},
  author={Bertsekas, Dimitri P},
  edition = {4th},
  volume={1},
  publisher={Athena Scientific},
  address={Belmont, MA},
  year={2020}
}

@article{douer1994optimal,
  title={Optimal Repair and Replacement in {Markovian} Systems},
  author={Douer, Nir and Yechiali, Uri},
  journal={Stochastic Models},
  volume={10},
  number={1},
  pages={253--270},
  year={1994},
  publisher={Taylor \& Francis}
}

@book{hernandez2012discrete,
  title={Discrete-time {Markov} Control Processes: {Basic} Optimality Criteria},
  author={Hern{\'a}ndez-Lerma, On{\'e}simo and Lasserre, Jean B},
  year={1996},
  address = {New York},
  publisher={Springer}
}

@article{fan2020optimal,
  title={Optimal Treatment of Chronic Kidney Disease With Uncertainty in Obtaining a Transplantable Kidney: {An} {MDP}-Based Approach},
  author={Fan, Wenjuan and Zong, Yang and Kumar, Subodha},
  journal={Annals of Operations Research},
  volume={316},
  number={11},
  pages={269--302},
  year={2020},
  publisher={Springer}
}

@article{alagoz2007choosing,
  title={Choosing Among Living-Donor and Cadaveric Livers},
  author={Alagoz, Oguzhan and Maillart, Lisa M and Schaefer, Andrew J and Roberts, Mark S},
  journal={Management Science},
  volume={53},
  number={11},
  pages={1702--1715},
  year={2007},
  publisher={INFORMS}
}

@article{bendersky2016deciding,
  title={Deciding Kidney-Offer Admissibility Dependent on Patients’ Lifetime Failure Rate},
  author={Bendersky, Michael and David, Israel},
  journal={European Journal of Operational Research},
  volume={251},
  number={2},
  pages={686--693},
  year={2016},
  publisher={Elsevier}
}

@article{kaufman2017living,
  title={Living-Donor Liver Transplantation Timing Under Ambiguous Health State Transition Probabilities},
  author={Kaufman, David and Schaefer, Andrew J and Roberts, Mark S},
  journal={Available at SSRN 3003590},
  year={2017}
}

\section*{AUTHOR BIOGRAPHIES}

\noindent {\bf XINGYU REN} is a Ph.D. student in the Department of Electrical and Computer Engineering at the University of Maryland, College Park. His research interests include stochastic optimization and Markov decision processes. His e-mail address is \email{renxy@umd.edu}.\\

\noindent {\bf MICHAEL C. FU} holds the Smith Chair of Management Science in the Robert H. Smith School of Business, with a joint appointment in the Institute for Systems Research and an affiliate appointment in the Department of Electrical and Computer Engineering, at the University of Maryland, College Park. His research interests include stochastic gradient estimation, simulation optimization, and applied probability. He served as WSC2011 Program Chair and received the INFORMS Simulation Society's Distinguished Service Award in 2018. He is a Fellow of INFORMS and IEEE. His e-mail address is \email{mfu@umd.edu}.\\

\noindent {\bf STEVEN I. MARCUS} is Professor Emeritus in the Department of Electrical and Computer Engineering and the Institute for Systems Research, University of Maryland.  He is former Editor-in-Chief of the SIAM Journal on Control and Optimization. His research is focused on stochastic control and estimation, Markov decision processes, and hybrid systems. He is a Fellow of IEEE and SIAM. His email address is \email{marcus@umd.edu}.\\

\end{document}